\documentclass[12pt]{article}
\usepackage{epsfig}

\usepackage{amsthm}
\usepackage{amsfonts}
\usepackage{amssymb}
\usepackage{graphicx}
\usepackage{amsmath}
\usepackage{lscape}
\usepackage{lscape,color}

\definecolor{c30}{rgb}{0,.7,.2}

\definecolor{d30}{rgb}{1,0,0}

\newcommand{\bea}{\begin{eqnarray}}
\newcommand{\eea}{\end{eqnarray}}

\newcommand{\be}{\begin{equation}}
\newcommand{\ee}{\end{equation}}

\renewcommand{\t}{\tau}
\newcommand{\tb}{\bar{\tau}}

\newtheorem{Theorem}{Theorem}[section]
\newtheorem{Definition}[Theorem]{Definition}

\newtheorem{Corollary}[Theorem]{Corollary}

\newtheorem{Lemma}[Theorem]{Lemma}

\begin{document}
\begin{small}\title{Three Dimensional Pure Gravity and Generalized Hecke Operators  }
\author{
M. Ashrafi\thanks{E-mail address: mashrafi@ucdavis.edu}\\[6pt]
 Department of Physics, University of California Davis \\One Shields Avenue, Davis, California 95616, USA }
\date{}
\maketitle
\begin{abstract}
In this paper, we study mathematical functions of relevance to pure gravity in AdS3. Modular covariance places stringent constraints on the space of such functions; modular invariance places even stronger constraints on how they may be combined into physically viable candidate partition functions. We explicitly detail the list of holomorphic and anti-holomorphic functions that serve as candidates for chiral and anti-chiral partition functions and note that modular covariance is only consistent with such functions when the left (resp. right) central charge is an integer multiple of 8, $c\in 8\mathbb{N}$. We then find related constraints on the symmetry group of the corresponding topological, Chern-Simons, theory in the bulk of AdS. The symmetry group of the theory can be one of two choices: either $SO(2; 1) \times SO(2; 1)$ or its three-fold diagonal cover. We introduce the generalized Hecke operators which map the modular covariant functions to the modular covariant functions. With these mathematical results, we obtain conjectural partition functions for extremal CFT2s, and the corresponding microcanonical entropies, when the chiral central charges are multiples of eight. Finally, we compute subleading corrections to the Beckenstein-Hawking entropy in the bulk gravitational theory with these conjectural partition functions
\end{abstract}
\end {small}
\section{Intoduction}
In three dimensional classical Einstein gravity, every two solutions are locally isometric and there is no propagating degrees of freedom. In 3d gravity with negative cosmological constant (AdS gravity), the existence of the BTZ black hole \cite{Banados:1992wn, Banados:1992gq} makes the theory more interesting to consider this theory as a toy model to understand the higher dimensional gravity \cite{Deser:1983tn}-\cite{Grumiller:2019ygj}.\\
  
 Three dimensional AdS gravity has asymptotic Virasoro symmetry. During the quantization of this theory, Brown-Henneaux showed that the theory has left and right moving Virasoro algebras which are  part of the structure of the conformal field theory \cite{Brown:1986nw}. The corresponding central charge is $c=\frac{3l}{2g}$ (where $l$ is the length of the AdS space).  This shows the existence of the boundary conformal field theory. This duality is an example of the AdS/CFT correspondence, which  is a correspondence between a bulk gravity, and boundary CFT in the lower dimension \cite{Maldacena:1997re}-\cite{Witten:1998qj}.
 
 Using the AdS/CFT dictionary,  one can obtain useful informations about the bulk theory by studying the boundary CFT. Solvability of the 3d gravity and the AdS/CFT correspondence, makes this theory more powerful to reveal some fundamental aspects of the quantum gravity\cite{Afkhami-Jeddi:2019zci}-\cite{Harvey:2020jvu}.

We  investigate the minimal theory of 3d gravity in this paper. It is believed that the 3d pure gravity is dual to the holomorphically factorizable extremal 2d CFT on the boundary, where there are some constraints on the central charges of the CFT \cite{witten}.  These constraints on the dual CFT come from the equivalence between 3d Einstein gravity and Chern-Simons gauge theory \cite{achucarro1986, witten1988cs}. 

In three dimension, the Einstein-Hilbert action with negative cosmological constant can be expressed in terms of the gauge invariant action as follows:
\begin{align}\label{E5}
 I=\frac{k}{4\pi}\int  {\rm Tr}^*\,(A\wedge dA+\frac{2}{3}A \wedge A \wedge A),
 \end{align}
where the gauge field $A$ is built from $SO(2,1)$ gauge field $w$(spin connection), and verbin $e$. The gauge group of this action is $SO(2,2)$. The gauge group $SO(2,2)$ is locally equivalent to $SO(2,1)\times SO(2,1)$. Therefore, in the oriented space-time the action can be written as:
\begin{align}\label{E9}
I&=\frac{k_L}{4 \pi}\int {\rm Tr}\, (A_L \wedge dA_L + \frac{2}{3}A_L \wedge A_L \wedge A_L) - \frac{k_R}{4 \pi} \int {\rm Tr}\,(A_R\wedge dA_R + \frac{2}{3} A_R\wedge A_R \wedge A_R)\\\nonumber
&=k_L I_L + k_RI_R.
\end{align}
where $k_L$ and $k_R$, are  Chern-Simons couplings.  The allowed values of the  $k_L$ and $k_R$, can be calculated from the quantization of the Chern-Simons couplings. Since the fundamental group of the gauge group $SO(2,1)\times SO(2,1)$ is $U(1)\times U(1)$, every diagonal covering groups of  this group can be considered as the gauge group of the Chern-Simons theory. For an n-fold diagonal  cover of $SO(2,1)\times SO(2,1)$, the quantization condition of the $k_L$ and $k_R$ is obtained as follows
\begin{equation}\label{klquanta}
k_L \in \bigg \{ \begin{array}{ll}
n^{-1} \mathbb{Z} ~~~~~~n \in 2\mathbb{Z}+1\\
(2n)^{-1}\mathbb{Z}~~~~~~n \in 2 \mathbb{Z}
\end{array} 
\end{equation}
\begin{equation}\label{kdifquanta}
k_L-k_R \in \mathbb{Z}
\end{equation}
Using the AdS/CFT dictionary, equation \eqref{E9} shows  the partition function of the dual CFT should be holomorphically factorizable:
\begin{equation}
Z(\tau,\bar{\tau})=Z(\tau)\bar{Z}(\bar{\tau})
\end{equation}
Where $Z(\tau)$ and $\bar{Z}(\bar{\tau})$ are chiral and anti chiral characters of the partition function which we called them as chiral and anti chiral functions respectively. For $c_L$ ( $c_R$ )integer multiple of 24, the chiral function  (anti-chiral function)  itself is a modular invariant partition function, and it can be  considered a corresponding purely chiral conformal field theory. Otherwise chiral function  is just the characters of the partition function \cite{flm1}.

The quantization condition \eqref{klquanta} leads to the following quantization condition on the left and right central charges:
\begin{equation}
(c_L,c_R)=(24k_L, 24k_R)
\end{equation}
 Quantization condition \eqref{kdifquanta} is equivalent to the $T$ invariant constraint  of the partition function of the dual CFT \cite{witten}.

  The purity of the 3d gravity  for the dual CFT  means that the primary fields of low dimensions should come from the identity. Therefore, the scaling dimension of the lowest primary fields excluding identity should be $k+1$. This class of CFTs is called extremal CFT \cite{hon1,hon2}. The vacuum state of the extremal CFT corresponds to the AdS space and the other primary fields  correspond to the BTZ black hole. \\
  Now, we are ready to solve the pure gravity. In order to achieve this goal, we need to find the dual CFT. First step toward obtaining the partition function of the extremal CFT is  determining the gauge group of the Chern-Simons gravity.  As   \eqref{klquanta} shows, the simplest gauge group is $SO(2,1)\times SO(2,1)$. For this gauge group, $k_L$ and $k_R$ are integer numbers and the dual extremal CFT satisfies all of the constraints. The left and right central charges are integer multiple of 24 and the chiral and anti-chiral  functions are modular invariant. 

Modular invariance  is a powerful tool that reveals interesting aspects of the conformal field theory \cite{heller1}-\cite{Cheng:2020srs}. Modularity determines the partition function precisely. Every modular function can be expressed in terms of the polynomial  of the Klein J function:
\begin{align}\label{wit4}
Z(q)=\sum_{r=0}^{k}f_rJ^r.
\end{align}
 For $c_L=24,48$, and $72,$  the partition function of the extremal CFT is calculated as follows:
\begin{align}\label{wit5}
Z_1(q)&=J(q)=q^{-1}+196884q+\cdots\\\nonumber
Z_2(q)&=J(q)^2-393767\\\nonumber
&=q^{-2}+1+42987520 q+\cdots\\\nonumber
Z_3(q)&=J(q)^3-590651J(q)-64481279\\\nonumber
&=q^{-3}+q^{-1}+1+2593096794 q+\cdots.
\end{align}
The closed-form of the extremal partition function is derived using the Hecke operators
 \begin{align}\label{He2}
T_nf(\tau)=\sum_{d|n}\sum_{b=0}^{d-1}f\left(\frac{n\tau +bd}{d^2}\right).
\end{align}
The Hecke operators map the modular functions to the modular functions and for $f(\tau)=q^{-1}+\mathcal{O}(q)$;
 \begin{align}\label{He2}
T_nf(\tau)=q^{-n}+\mathcal{O}(q).
\end{align}
Therefore, the partition function of external CFT is obtained as follows \cite{appostal}\footnote{One can also write the closed-form of the partition function in terms of the unique modular function $J_m(\tau)$ which only has an order-m pole at $q=0$ \cite{McGady:2018rmo}:
\begin{align}
Z_k(\tau) := J_k(\tau) + \sum_{m = 0}^{k-1}  \{ p(k-m) - p(k-m-1) \}J_m(\tau),
\end{align}
where $p(m)$ is the partition number.}:

\begin{align}\label{wit6}
Z_k(\tau)=\sum_{r=0	}^{k}a_{-r}T_rJ(\tau).
\end{align}
Where,  $a_{-r}$ are the coefficients of the low states of the vacuum. For $c_L=24$, it is believe that  there exist $71$ holomporphic CFT \cite{flm1}.  70 of these holomorphic CFTs have  Kac-Moody or current algebra symmetry. Therefore, they can not be a candidate for pure gravity. The partition function $Z_1$ is one of the 71 holomorphic CFTs with no Kac-Moody symmetry. This model first was constructed by Frenkel, Lepowsky, and Meurman and its uniqueness was conjectured  \cite{ flm2}. For $c>24$ it is not know whether such CFTs exist or not. Although  the existence of extremal CFT were investigated for large value of $k$, but its existence  is still an open question \cite{gaber1, Gaberdiel:2008pr}. If such CFT's exist, they are good candidates for pure gravity.

In \cite{Maloney:2007ud, Keller:2014xba}, the authors attempted to compute the partition function of the pure gravity form different point of view by summing over the known saddles in the Euclidean gravitational path integral. In \cite{Alday:2019vdr}, using the Rademacher expansiosn  the partition function of the pure gravity have been calculated. The resulting partition function is an interesting modification of the partition function in \cite{Maloney:2007ud, Keller:2014xba}. Their results showed some unphysical features  (e.g., the negativity of the density of states at special values of the primary fields). Despite the attempts to solve these issues the  partition function of the pure gravity is still unknown \cite{Benjamin:2019stq}-\cite{Maxfield:2020ale}.

In \cite{witten}, Witten considered the $SO(2,1)\times SO(2,1)$ gauge group, now the question arise: ``Does any covering group of the gauge group $SO(2,1)\times SO(2,1)$ exists which satisfies all constraints on the dual CFT?'' 

In order to address this question, we  study the holomorphically  factorizable CFT  in this paper. We show that modular invariance  of    the holomorphically factorizable CFT  is  necessary and sufficient condition for deriving the allowed values of the covering group. In section \ref{s2}, we  study the modularity of the holomorphically factorizable partition functions and we  calculate the allowed covering group. From modular invariance of the holomorphically factorizable partition function we conclude that  the chiral  and the anti-chiral  functions are modular covariant. We explicitly detail the list of holomorphic and anti-holomorphic functions that serve as candidates for chiral and anti-chiral partition functions and note that modular covariance is only consistent with such functions when the left (resp. right) central charge is an integer multiple of 8, $c\in 8\mathbb{N}$.  Since the chiral and the anti-chiral  functions are not modular invariant, we can not use the Hecke operators. In section \ref{sec3}, we  introduce the generalized Hecke operators which map modular covariant functions into the modular covariant functions. We also investigate its Fourier expansion. In section \ref{sec4}, we  find related constraints on the symmetry group of the corresponding topological, Chern-Simons, theory in the bulk of AdS. We show that the symmetry group of the theory can be one of two choices: either $SO(2; 1) \times SO(2; 1)$ or its three-fold diagonal cover. In this section, we study the dual CFT for the case where the gauge group is three-fold diagonal cover of the group $SO(2,1)\times SO(2,1)$. we obtain conjectural partition functions for extremal CFT2s, and the corresponding microcanonical entropies, when the chiral central charges are multiples of eight.
 \section{Holomorphically Factorizable Partition Function}\label{s2}
 \subsection{Partition Function}\label{sub21}
 The partition function of the unitary 2d CFT in the upper half plane $\tau=\tau_r+i\tau_i$ ($\tb=\tau_r-i\tau_i$), is defined as 
  \begin{equation}
    \label{z1}
    Z(\tau,\bar{\tau})=q^{\frac{-c_L}{24}}{\bar q}^{\frac{-c_R}{24}}\sum_{h,\bar{h}=0}\rho(h,\bar{h})\,q^{h}{\bar q}^{\bar{h}},
    \end{equation}
  where, $c_L$ and $c_R$ are the left and right central charges and $\rho(h,\bar{h})$ is the density of the state.

The holomorphically factorizable partition functions can be written as the multiplications of the chiral and anti-chiral  functions as follows:
\be
Z(\t,\tb)=Z(\t)\bar{Z}(\tb),
\ee
where the chiral function $Z(\tau)$,  and the anti-chiral  function $\bar{Z}(\tb)$ are defined  as follows
\be\label{HF2}
Z(\t)=\sum_{h=0}\rho(h)e^{2\pi i\t(h-\frac{c_L}{24})},
\ee
\be\label{HF3}
\bar{Z}(\bar{\tau})=\sum_{\bar{h}=0}\rho(\bar{h})e^{-2\pi i\tb(\bar{h}-\frac{c_R}{24})}.
\ee
\subsection{Modularity of Partition Function}\label{sub22}
The modular covariance of  the chiral and anti-chiral  functions is the necessary and sufficient condition for the  modular invariance of the partition function $Z(\t ,\tb)$. Modular covariant  means that the chiral and anti-chiral  functions take  an overall phase under  the $S$ and $T$ transformations. 
 
Modular covariance of $Z(\t)$ under the $S$ transformation demands:
\be\label{HF4}
Z\left(\frac{-1}{\t}\right)=e^{i\beta}Z(\t).
\ee
The identity $S^2=1,$ shows  the phase $\beta$ should be $\pi$ or $2\pi$. 

For  $\beta=\pi$, we called the corresponding  function $Z^{-}(\t)$. The Fourier expansion of $Z^{-}(\t)$ under $S$ transformation at the self dual point $\t=i$ yields:
\be\label{HF5}
SZ^{-}(\t)\big\vert_{\t=i}=\sum_{h=0} \rho(h) e^{-2\pi(h-\frac{c_L}{24})}.
\ee
Using \eqref{HF4} for $\beta=\pi$, and  \eqref{HF5} show that $Z^{-}(\t=i)$ is equal to  zero. In Eq. \eqref{HF5} all phases are positive, so some of the density of states should be negative. Therefore, $Z^{-}(\t)$ is not a  physical function.
 
Covariance of $Z(\t)$ under $T$ transformation requires:
\be\label{HF6}
TZ(\t)=e^{-2\pi i\alpha}Z(\t).
\ee
Plugging \eqref{HF2} into \eqref{HF6}, for $\tau_r=0$ yields
\be\label{HF7}
\sum_{h=0}\rho(h){e}^{-2\pi\tau_i(h-\frac{c_L}{24})}\left(1-\cos2\pi (h-\frac{c_L}{24}+\alpha)\right)=0,
\ee
\be\label{HF8}
\sum_{h=0}\rho(h){e}^{-2\pi\tau_i(h-\frac{c_L}{24})}\sin2\pi \left(h-\frac{c_L}{24}+\alpha\right)=0.
\ee
The summands  \eqref{HF7} and \eqref{HF8} are non-negative. Therefore, $\left(h-\frac{c_L}{24}+\alpha\right)$ should be integer. The vacuum state $(h=0)$, requires that
\be\label{alpha}
 \alpha=\frac{c_L}{24},
 \ee
Therefore, 
\be\label{h}
h\in\mathbb{N}.
\ee
Using $(ST)^3=1$ and  invariance of $Z(\t)$ under $S$ transformation one can obtain:
\be\label{HF10}
e^{-2\pi i\frac{c_L}{8}}=1.
\ee
Consequently;
\be\label{c8}
c_L=8m_L,\quad m_L\in\mathbb{N}
\ee
 Similarly, for $\bar{Z}(\tb)$ we have:
\begin{eqnarray}\label{HFC2}
\bar{h}\in\mathbb{N}\quad\quad c_R=8m_R,\quad m_R\in\mathbb{N}.
\end{eqnarray}
 For $m_L, m_R \notin3\mathbb{N}$ modular invariance of the partition function $Z(\t,\tb)$, enforces that $m_L=m_R=k$. In this case, the partition function $Z(\t,\tb)$ automatically becomes real. For $m_L, m_R \in3\mathbb{N}$,  if we put the weak condition of the reality of the partition function, this constrains leads to the equality of the right and left central charges.
\subsection{The Basis for the Modular Covariant  Functions}\label{sub32}
In this section we  derive the basis for $Z(\t)$. For $m_L$ integer multiple of  three, i.e. $c_L\in 24\mathbb Z$, $Z(\t)$ is modular invariant.  Therefore, it is a polynomial  in terms of the Klein function $J(\t)$ \cite{appostal}:
\be
      Z(\tau)=\sum_{r=0}^{k}h_r J^{r}.
     \ee
    with some coefficients $h_r$.  The Klein function has the following Fourier expansion:
    \begin{eqnarray}
    \label{g1}
    J(\t)&=&\mathfrak{j}(\t)^3\\
    \label{j11}
    &=&q^{-1}+744+196884\, q+\cdots,
    \end{eqnarray}
    where, $\mathfrak{j}$ has expansion in terms of the Jacobi Theta functions and Eta function as follows 
\begin{eqnarray}
    \label{g11}
    \mathfrak{j} (\t)&:=&\frac{1}{2}\left[ \left(\sqrt{\frac{\theta_2(\t)}{\eta(\t)}}\right)^{16}+\left(\sqrt{\frac{\theta_3(\t)}{\eta(\t)}}\right)^{16}+\left(\sqrt{\frac{\theta_4(\t)}{\eta(\t)}}\right)^{16}\right] \nonumber\\
     &=&q^{\frac{-1}{3}}\left(1+248\,q+\cdots\right).
      \end{eqnarray}
   In order to obtain the bases for $Z(\t)$, we use the lemma  in \cite{Ashrafi-16}.    
    \begin{Lemma}\label{t41}
 The $S$-invariant function $f^{(r)}\left(\lbrace a^{(r)}\rbrace,\t\right)$  with  Fourier expansion
    \begin{align}\label{g31}
    &f^{(r)}\left(\lbrace a^{(r)}\rbrace,\t\right)=q^{\frac{-p}{3}}\left[\sum _{n=-r}^{0} a_n^{(r)} q^n+\sum _{n=1}^{\infty } a_n^{(r)} q^{n}\right],&p\in\{0,1,2\}.
    \end{align}
    on the upper half  $\tau$-plane, is a polynomial in $\mathfrak{j}$.
    \end{Lemma}
       \begin{proof}
  From Eq.\eqref{j11} and Eq.\eqref{g11} we conclude  that there exist $\lbrace a^{(r-1)}\rbrace$ such that
    \be\label{g311}
    q^{\frac{-p}{3}}\sum _{n=-r}^{\infty} a_n^{(r)} q^n=a_{-r}^{(r)}\,\mathfrak{j}^p\, J^r+q^{\frac{-p}{3}}\sum _{n=-r+1}^{\infty} a_n^{(r-1)} q^n.
    \ee
    Therefore,
    \be\label{g312}
    f^{(r)}\left(\lbrace a^{(r)}\rbrace,\t\right)=a_{-r}^{(r)}\,\mathfrak{j}^p\, J^{r}+f^{(r-1)}\left(\lbrace a^{(r-1)}\rbrace,\t\right).
    \ee
    Since $f^{(r)}\left(\lbrace a^{(r)}\rbrace,\t\right)$, $\mathfrak{j}$ and $J$ are  $S$-invariant, so $f^{(r-1)}\left(\lbrace   a^{(r-1)}\rbrace,\t\right)$ is also $S$-invariant.
   The order of the pole of $f^{(r)}\left(\lbrace a^{(r)}\rbrace,\t\right)$ and $f^{(r-1)}\left(\lbrace a^{(r-1)}\rbrace,\t\right)$ are  $r$ and $r-1$ respectively. By iteration one can obtain
    \be\label{g313}
    f^{(r)}\left(\lbrace a^{(r)}\rbrace,\t\right)=\mathfrak{j}^p\left[a_{-r}^{(r)}J^{r}+a_{-(r-1)}^{(r-1)}J^{r-1}+\cdots +a_{0}^{(0)}\right]+f^{(-1)}\left(\lbrace a^{(-1)}\rbrace,\t\right),
    \ee
    where
    \be\label{g32}
    f^{(-1)}\left(\lbrace a^{(-1)}\rbrace,\t\right)=q^{1-\frac{p}{3}}\sum_{m\ge1} a^{(-1)}_m q^{m-1}.
    \ee
   Since the function $f^{(r)}\left(\lbrace a^{(r)}\rbrace,\t\right)$ is $T^3$ invariant for all values of $r$,  therefore  the function $\left[f^{(-1)}\left(\lbrace a^{(-1)}\rbrace,\t\right)\right]^3$ is modular invariant. It has no pole on the upper half plane and is zero at $\tau=i\infty$. Thus, it is zero on the upper half plane.
    \end{proof}
  \begin{Corollary}\label{C3}
    The  chiral  function  is a polynomial in function $\mathfrak{j} $ as follows:
     \begin{align}
     \label{ch1}
 Z(\t)=\mathfrak{j}^{k}\sum_{r=0}^{[k/3]}n_r J^{-r}, \quad &n_r\in {\mathbb N}.
    \end{align}
    \end{Corollary}






\section{The Hecke Operators}\label{sec3}
 Let us define the subgroup of the modular group with $S$ and $T^3$ generators, which we call this group  $\Gamma_3$. 

The modular covariant function $f_3(\tau)$ with the Fourier expansion:
 \begin{equation}
 f_3(\tau)=\sum_{m=-k}^{\infty}a\left(m-\frac{1}{3}\right)q^{m-\frac{1}{3}},
 \end{equation}
is invariant under the group $\Gamma_3$. 

The Hecke operators are  linear operators which map modular form space $M_k$, onto itself and are defined as follows:
\be\label{h0}
T_nf(\tau)=n^{k-1}\sum_{d|n}d^{-k}\sum_{b=0}^{d-1}f\left(\frac{n\tau +bd}{d^2}\right).
\ee
The Hecke operators map the Modular functions, onto the modular functions. \\

  In this section, we  generalize  the definition of the Hecke operators, 3-Hecke operators for the group $\Gamma_3$ , which map the modular covariant functions onto the modular covariant functions. \\
 \begin{Definition}\label{def1}
 For positive integer values of $n$ and $n\neq 3\mathbb{N}$, the operator $T^{(3)}_n$ on the modular covariant function $f_3(\tau)$ is defined as follows
   \begin{equation}\label{he3}
 T^{(3)}_nf_3(\tau)=\sum_{d|n}\sum_{b=0}^{d-1}f_3\left(\frac{n\tau +3bd}{d^2}\right),\quad\quad  \end{equation}
\end{Definition}
We called the operators $T^{(3)}_n$, the 3-Hecke operators. We show that the 3-Hecke operators map the modular covariant functions  $f_3(\tau)$,  onto the modular covariant functions.
First, we study the Fourier expansion  of  $T^{(3)}_nf_3(\tau)$.
\begin{Theorem}
If $f_3(\t)$ has the  Fourier expansion
 \begin{equation}\label{fourierf3}
 f_3(\tau)=\sum_{m=-k}^{\infty}a\left(m-\frac{1}{3}\right)q^{m-\frac{1}{3}}.
 \end{equation}
 then, $T^{(3)}_nf_3(\tau) $ has the  Fourier expansion:
\be\label{fourierheck}
T^{(3)}_nf_3(\tau)=\sum_{m=-k}^{\infty} \gamma_n \left(m-\frac{1}{3}\right) q^{m-\frac{1}{3}},
\ee
where
\be\label{h2}
\gamma_n (m) =\sum_{d|(n,3m)}\frac{n}{d}a\left(\frac{nm}{d^2}\right).
\ee
\end{Theorem}
\begin{proof}
By putting the Fourier expansion of the function $f_3(\tau)$\eqref{fourierf3} into \eqref{he3} we have
\begin{equation}\label{fourier1}
T^{(3)}_nf_3(\tau)=\sum_{m=-k}^{\infty}a\left(m-\frac{1}{3}\right)\sum_{d|n}e^{\frac{2\pi i n\tau}{d^2}\left(m-\frac{1}{3}\right)}\sum_{b=0}^{d-1}e^{\frac{2\pi ib(3m-1)}{d}}.
\end{equation}
The last sum in \eqref{fourier1} is zero for $d\nmid 3m-1$ and is equal to $d$ for $d\mid 3m-1$:
\begin{equation}\label{fourier2}
T^{(3)}_nf_3(\tau)=\sum_{m=-k}^{\infty}a\left(m-\frac{1}{3}\right)\sum_{d|n, d\mid 3m-1}de^{\frac{2\pi i n\tau}{d^2}(m-\frac{1}{3})}.
\end{equation}
Since, $d\mid 3m-1$; writing  $3m-1=pd$ and replacing $\frac{n}{d}$ with $d$ (because $d\mid n$) yeilds
\begin{equation}\label{fourier3}
T^{(3)}_nf_3(\tau)=\sum_{p=-3k-1}^{\infty}\sum_{d|n}a\left(\frac{pn}{3d}\right)\frac{n}{d}e^{\frac{2\pi i \tau dp }{3}}.
\end{equation}
The last term in the sum has the form $q^({\frac{pd}{3})}$. For all  terms which ${\frac{pd}{3}}$ is constant ${\frac{pd}{3}}=m-\frac{1}{3}$ one can obtain:  
\be\label{fourierheck1}
T^{(3)}_n f_3(\tau)=\sum_{m=-k}^{\infty}\sum_{d|(n,3m-1)}\frac{n}{d}a\left((m-\frac{1}{3})\frac{n}{d^2}\right) e^{2\pi i\tau(m-\frac{1}{3})}.
\ee
\end{proof}
%
%
\subsection{The Order $n$ Transformations}\label{sub31}
\indent For positive integer $n$, the order $n$ transformation $\Gamma(n)$, is defined as follows 
\begin{equation}
\tau\rightarrow A\tau=\frac{a\tau+b}{c\tau+d},\quad\quad ad-bc=n,
\end{equation}
 where  $a,b,c$, and $d$ are integers. The $\Gamma (1)=\Gamma$ transformations  correspond to the modular transformations.\\
 \indent The transformations $A_1$ and $A_2$ in  $\Gamma (n)$ are called equivalent if there exist a modular transformation $V\in\Gamma$, such that
 \begin{equation}
A_2\sim  A_1\quad \mbox{if}\quad A_2=VA_1.
 \end{equation}
 It is clear that the relation $\sim$ is an equivalence relation. So, the transformations $\Gamma (n)$  can be divided into the equivalence classes. Two element of $\Gamma (n)$ are in the same class, if and only if, they are equivalent.\\

\begin{Lemma}\label{lem2}
For every equivalence class of $\Gamma (n)$, there is a triangular representation $A_3$:
 \begin{equation}
 A_3=\left(\begin{array}{cc} a & 3b \\ 0 & d \\ \end{array} \right),
\end{equation}
where $n=3p+i$$(i=1,2)$ and $p\in\mathbb{N}$.
\end{Lemma}
\begin{proof}
As shown in \cite{appostal}, in every equivalence class of $\Gamma(n)$ there is a representation of triangular form
 \begin{equation}
 A_1=\left(\begin{array}{cc} a_1 & b_1 \\ 0 & d _1\\ \end{array} \right)\quad\quad d_1>0.
\end{equation}
 For $A_1$ and $A_2$ (two equivalent elements in $\Gamma (n)$), there is $V=\left(\begin{array}{cc} 1 & q \\ 0 & 1 \\ \end{array} \right)\in\Gamma$
 Such that
\be
 A_2=VA_1=\left(\begin{array}{cc} a_1 & qd_1+b_1 \\ 0 & d_1 \\ \end{array} \right).
 \ee
In order to prove this theorem, it  is necessary to show that $b_2$ is multiple of three.
\be\label{1}
b_2=qd_1+b_1.
 \ee
 Since $a_1d_1=n$ and $n\neq 3\mathbb{N}$, so $d_1$ can not be multiple integer of three and takes   $3s+1$ or $3s+2$ values. For fixed value of d, $b_1$  takes  $3r, 3r+1$ and $3r+2$  values. By substituting these values to \eqref{1}, one can show $b_2$ can be multiple of three (by choosing  appropriate values of $q$).
\end{proof}
\begin{Theorem}\label{theorem4}
A complete system of nonequivalent elements of  $\Gamma (n)$ is given by the set of triangular transformations of the form:
 \begin{equation}\label{2}
\left(\begin{array}{cc} a & 3b \\ 0 & d \\ \end{array} \right),
\end{equation}
 where $d$ runs through the positive divisors of n and for fixed values of $d$, $a=\frac{n}{d}$ and $b$ runs through a complete residue system of modulo $d$.\\
\end{Theorem}
\begin{proof}
The  lemma \eqref{lem2} shows every element of $\Gamma (n)$ is equivalent to one of the transformations in \eqref{2}. So, we should show  two  transformations $A_1$ and  $A_2$  are equivalent, if and only if
  \be\label{pr3}
  a_1=a_2, \quad d_1=d_2, \quad \mbox{and}\quad  b_2=b_1+q'd.
  \ee
  where,
   \begin{equation}
 A_i=\left(\begin{array}{cc} a_i & 3b_i \\ 0 & d_i \\ \end{array} \right) \quad\quad i=1,2.
\end{equation}
 First, we show if \eqref{pr3} holds, then $A_1\sim A_2$. For some integer $q$, if we consider $V$ as follows
 \be
 V=\left(\begin{array}{cc} 1 & q \\ 0 & 1 \\ \end{array} \right)\in\Gamma.
 \ee
 where $q=3q'$, then, $VA_1=A_2$, so $A_1\sim A_2$.\\
 Conversely, if $A_1\sim A_2$ there exists $V\in \Gamma$
 \be
V=\left(\begin{array}{cc} p & q \\ r & s \\ \end{array} \right),
 \ee
 such that
 \be
\left(\begin{array}{cc} a_2 & 3b_2 \\ 0 & d_2 \\ \end{array} \right)=\left(\begin{array}{cc} p & q \\ r & s \\ \end{array} \right) \left(\begin{array}{cc} a_1 & 3b_1 \\ 0 & d_1 \\ \end{array} \right)= \left(\begin{array}{cc} pa_1 & 3pb_1+qd_1 \\ ra_1 & 3rb_1+sd_1 \\ \end{array} \right).
 \ee
 The above equality shows $r=0$ (since $a_1\neq 0$). From $ps-qr=1$, we can conclude  $ps=1$, so $p=s=1$ or $p=q=-1$. Let us consider $p=s=1$ (for the other case we replace $V$ by $-V$).  By equating the entries in the above equation, we have
  \be\label{3}
  a_1=a_2,\quad  d_1=d_2, \quad and\quad  3b_2=3b_1+qd.
  \ee
  Since, $a_1d_1=n$ and $n$ is not integer multiple of three, therefore, $q=3q'$:
    \be\label{3}
 b_2=b_1+q'd.
  \ee
\end{proof}

 \begin{Lemma}\label{lem5}
 For $A_1\in \Gamma (n)$, $V_1\in \Gamma$ there exists transformation $A_2\in\Gamma(n)$ and $V_1\in \Gamma$ such that
 \be
A_1V_1=V_2A_2.
\ee
where
  \be
 A_i=\left(\begin{array}{cc} a_i & 3b_i \\ 0 & d_i \\ \end{array} \right), \quad and \quad\quad V_i=\left(\begin{array}{cc} \alpha _i & \beta _i \\ \gamma _i & \delta _i \\ \end{array} \right).
 \ee
 \end{Lemma}
\begin{proof}
$det(A_1V_1)=detA_1detV_1 =n$, so $A_1V_1\in \Gamma (n)$. According to lemma \ref{lem2}, there exists  $A_2\in \Gamma (n)$ and $V_2\in \Gamma$ such that
 \be\label{4}
A_1V_1=V_2A_2.
\ee
\end{proof}

 Now, by using \eqref{4} for $V_1=S=\left(\begin{array}{cc} 0 & -1 \\ 1 & 0 \\ \end{array} \right)$ transformation, we derive  the elements of $A_2$ and $V_2$ in terms of the element of $A_1$. By equating the entries in \eqref{4}, we have
 \begin{eqnarray}\label{5}
 a_2&=&\frac{d_1}{\gamma_2},\\\nonumber
  3b_2&=&-a_1\delta_2, \\\nonumber
   d_2&=&a_1\gamma_2,\\\nonumber
  \alpha_2 &=&\frac{3b_1\gamma_2}{d_1},\\\nonumber
   \beta_2 &=&\frac{3b_1\delta_2}{d_1}-\frac{1}{\gamma_2}.\\\nonumber
 \end{eqnarray}
From \eqref{5}, we recognize that  $V_2$ has two independent entries $\delta_2$ and $\gamma_2$. Since; $n=a_1d_1\neq 3\mathbb{N}$ , the second and the forth equation in \eqref{5} show  $\delta_2$ and $\alpha_2$ are multiples of three. \\
\indent We already know  $S$ and $T=\left(\begin{array}{cc} 1 & 1 \\ 0 & 1 \\ \end{array} \right)$ transformations are generators of the modular group and each elements of the modular group can be written in the below form
\be
ST^{n_1}ST^{n_2}.......
\ee
 Since $V_2$ has two independent entries, one can write it as follows
  \be\label{6}
V_2=ST^{n_1}ST^{n_1}ST^{n_2}.
 \ee
 From \eqref{6} we have
 \begin{eqnarray}\label{7}
  \alpha_2 &=&\frac{3b_1\gamma_2}{d_1}=-n_1,\\\nonumber
   \beta_2 &=&\frac{3b_1\delta_2}{d_1}-\frac{1}{\gamma_2}=1-n_1n_2,\\\nonumber
   \gamma _2&=&n_1^2-1,\\\nonumber
   \delta_2 &=& n_{2}(n_1^2-1)-n_1.
 \end{eqnarray}
 For the case where $n$ is not multiple of three, from \eqref{5} and \eqref{7} we conclude  $n_1$ and $n_2$ are multiple of three.\\

\begin{Theorem}
 For  an integer value of  $n$ which is not multiple of three, if $f_3(\tau)$ is modular covariant then, $T^{(3)}_n f_3(\tau)$ is covariant under modular transformations. Hence, $T^{(3)}_n f_3(\tau)$ maps the $\Gamma_3$ invariant function $f_3(\t)$ onto the $\Gamma_3$ invariant function.
 \end{Theorem}
\begin{proof}
 Since $d|n$ one can  rewrite the 3-Hecke operator as follow
 \be\label{he1}
T^{(3)}_nf_3(\tau)=\sum_{a\geq 1, ad=n}\sum_{b=0}^{d-1}f_3(A\t),
\ee
where $A$ is an element of $\Gamma (n)$:
\begin{equation}
A\tau=\frac{a\tau+3b}{d}.
\end{equation}
 From \eqref{he1} we have:
  \be\label{he01}
T^{(3)}_n f_3(S\tau)=\sum_{a_1\geq 1, a_1d=n}\sum_{b=0}^{d-1}f_3(A_1S\t),
\ee
Using lemma \eqref{lem5}, we have
  \be\label{he21}
f_3(A_1S\t)=f_3(ST^{n_1}ST^{n_1}ST^{n_2}A_2\t)=e^{\frac{-2(2n_1+n_2)\pi i}{3}}f_3(A_2\t).
\ee
 As we showed earlier,  $n_1$ and $n_2$ are multiple of three, therefore:
  \be\label{heck44}
f_3(A_1S\tau)=f_3(A_2\tau).
 \ee
 Substituting \eqref{heck44} to \eqref{he1} yields to:
  \be\label{he12}
T^{(3)}_n f_3(S\tau)=\sum_{a_2\geq 1, a_2d=n}\sum_{b=0}^{d-1}f_3(A_2\t)=T^{(3)}_n f_3(\tau).
\ee
\eqref{he12} shows that the 3-Hecke operators are invariant under $S$ transformation. The Fourier expansion \eqref{fourierheck}, shows that the 3-Hecke operators are modular covariant under $T$ transformation and are invariant under $T^3$ transformation. Since every elements of the group of $\Gamma_3$ are built from the multiplication of the $S$ and $T^3$ generators, we conclude that the 3-Hecke operators map the $\Gamma_3$ invariant functions $f_3(\t)$ into the $\Gamma_3$ invariant functions.
\end{proof}
\section{Three Dimensional Gravity}\label{sec4}
Our focus in this section is solving the pure quantum gravity in the sense of  finding the dual boundary CFT. As  it is shown in \cite{witten}, the dual CFT is extremal which means that the lowest dimension of the primary fields excluding the identity, is $k+1$ for $c=24k$, and the partition function should be holomorphically factorizable. The allowed values of the left  and right central charges can be obtained from the symmetry group of the Chern-Simons gauge theory. The symmetry group can be the group $SO(2,1)\times SO(2,1)$ and its $n$th diagonal cover:
\begin{equation}\label{klkr2}
k_L \in \bigg \{ \begin{array}{ll}
n^{-1} \mathbb{Z} ~~~~~~n \in 2\mathbb{Z}+1\\
(2n)^{-1}\mathbb{Z}~~~~~~n \in 2 \mathbb{Z}
\end{array} 
\end{equation}
where,  $k_L$ and $k_R$ are the Chern-Simons couplings. From the AdS/CFT dictionary, the corresponding central charges are obtained as follows
\begin{equation}
(c_L,c_R)=(24k_L, 24k_R)
\end{equation}
Therefore, the first step to solve the quantum gravity is determining  the covering group. The holomorphically factorizable and modular invariance of the partition function is necesary and sufficient condition for determining the symmetry group. 

In section \eqref{s2} we show  for the modular invariant holomorphically factorizable partition function, the scaling dimensions of the chiral and anti-chiral  functions should be an integer number and the left and right central charges are integer multiples of eight:
\begin{equation}\label{clcr}
c_L=8m_L, \quad\quad c_R=8m_R,\quad\quad\quad\quad m_L,m_R\in\mathbb{N}.
\end{equation}
From \eqref{klkr2} and \eqref{clcr}, we conclude that the allowed values of  $n$ should be $1$ or $3$. So,  the symmetry group of the Chern-Simons gauge theory should be $SO(2,1)\times SO(2,1)$ and its three-fold diagonal cover.

For pure gravity with the asymptotic space-time  AdS$_3$, the  vacuum state is the trivial state where its chiral  function is obtained as follows
\begin{equation}\label{zvac}
Z^{vac}(\tau)=q^{-\frac{c_L}{24}}\prod_{n=2}^{\infty}\frac{1}{1-q^n},
\end{equation}
The  vacuum state   corresponds to the Anti de Sitter Space, classically. Since, the vacuum partition function \eqref{zvac}, is not modular covariant, there should be other states in the theory. This is in the agreement with the existence of the BTZ black hole in the theory. The mass and the angular momentum of the classical BTZ black hole in terms of the Virasoro generators $L_0$ and $\bar{L}_0$ are obtained as follows
\begin{eqnarray}\label{mass}
M&=&\frac{1}{l}(L_0+\bar{L}_0),\\\nonumber
J&=&(L_0-\bar{L}_0).
\end{eqnarray}
where
\be\label{mj}
Ml\geqslant |J|.
\ee
and the  entropy is
\be\label{entropy0}
S=4\pi\left(\sqrt{\frac{c_L}{24}L_0}+\sqrt{\frac{c_R}{24}\bar{L}_0}\right).
\ee
From \eqref{mj} and \eqref{entropy0}, we conclude that $L_0\geqslant1$. Hence, the full chiral function has the following form:
\begin{equation}\label{ztot}
Z(\tau)=q^{-\frac{c_L}{24}}\prod_{n=2}^{\infty}\frac{1}{1-q^n}+{\cal O}(q).
\end{equation}
The modular covariant constraint determines the partition function uniquely. For $k_L, k_R\in\mathbb{Z}$, the chiral and anti-chiral  functions are modular functions. In \cite{witten}, the partition function and the entropy are investigated. In this section we  study the case where  $k_L, k_R\in\frac{\mathbb{Z}}{3}$. We showed in section \eqref{s2}, modular invariant of $Z(\t,\tb)$ demands $c_L=c_R=8k, k\in\mathbb{Z}$.

From corollary \eqref{C3}, the chiral  function is obtained as follows 
 \begin{align}
     \label{ch11}
 Z(\t)=\mathfrak{j}^{k}\sum_{r=0}^{[k/3]}n_r J^{-r},\quad &n_r\in {\mathbb N}.
    \end{align}
    where the $n_r$  coefficients are determined from the fact that the density of the low dimensional state, should be equal to the density of the corresponding low dimensional state of the vacuum. It is worth noting that for $k=3m$ for some $m\geq 1$, these candidate partition functions are exactly the candidate torus path integral  introduced in \cite{witten} . For $k=1$ to $k=11$, the chiral  functions is obtained in \cite{Avramis:2007gx}. Here are some examples:
    \begin{align}\label{wit14}
Z_8(\tau)=\mathfrak{j}(\tau)=q^{\frac{-1}{3}}+248q^{\frac{2}{3}}+4125q^{\frac{5}{3}}+\cdots,
\end{align}
\begin{align}\label{wit15}
Z_{16}(\tau)=\mathfrak{j}^{2}(\tau)=q^{\frac{-2}{3}}+496q^{\frac{1}{3}}+69752q^{\frac{4}{3}}\cdots,
\end{align}
\begin{align}\label{wit16}
Z_{32}(\tau)&=\mathfrak{j}^{4}(\tau)-992\mathfrak{j}(\tau)\\\nonumber
&=q^{\frac{-4}{3}}+139504q^{\frac{1}{3}}+69332992q^{\frac{5}{3}}\cdots,
\end{align}
\begin{align}\label{wit17}
Z_{40}(\tau)&=\mathfrak{j}^{5}(\tau)-1240\mathfrak{j}^{2}(\tau)\\\nonumber
&=q^{\frac{-5}{3}}+20620q^{\frac{1}{3}}+86666240q^{\frac{4}{3}}\cdots.
\end{align}
The chiral function \ref{ch11} is unitary modular covariance and has the positive density of state. For $c=8$ and $c=16$, the chiral functions $Z_8$ and $Z_{16}$ are well known and is believed that they are unique. The chiral functions $Z_8$ and $Z_{16}$ are the vacuum character of the level 1 affine $\hat{E}_8$ and the level 1 affine $\hat{E}_8\times\hat{E}_8$ theory respectively. For $c=32$  and $c=40$ the chiral function has been identified with $\mathbb{Z}_2$ orbifolds of theories defined on even unimodular lattices of the respective rank possessing no vectors of squared length 2. For $c>40$, the existence of these CFTs are not known and is an open question \cite{hon1, hon2}. If these CFT's exist they are good candidates for the pure gravity.

The closed-form of the chiral function is obtained by using the generalize 3-Hecke operator \eqref{he3}. The Fourier expansion \eqref{fourierf3} and \eqref{fourierheck} shows for 
   \begin{align}\label{j1}
\mathfrak{j}(\tau)=q^{\frac{-1}{3}}+{\cal O}(q).
\end{align}
and
\begin{align}\label{j2}
\mathfrak{j}^{2}(\tau)=q^{\frac{-2}{3}}+{\cal O}(q).
\end{align}
The 3-Hecke operators have  the following expansion
   \begin{align}\label{he3j1}
T^{(3)}_{n}\mathfrak{j}(\tau)=q^{\frac{-n}{3}}+{\cal O}(q).
\end{align}
and
\begin{align}\label{he3j2}
T^{(3)}_{n}\mathfrak{j}^{2}(\tau)=q^{\frac{-2n}{3}}+{\cal O}(q).
\end{align}
Therefore, for $c=8k$   the chiral  function is obtained as follows
\be\label{z2k+1}
Z_{k}(\tau)=\sum_{r=0}^{k}a_{-r}T^{(3)}_r\mathfrak{j}^{i}(\tau),\quad\quad i=1,2.
\ee
where $i=1$ is for $k=\mbox{odd}$  and $i=2$ corresponds with $k=\mbox{even}$. The  $a_r$ coefficients, are the low state density of the vacuum:
\be
Z^{vac}(\t)=\sum_{r=-k}^{\infty}a_rq^r.
\ee
In order to determine the entropy, let us write the chiral function as:
\be\label{partition11}
Z_k(\t)=\sum_{m=-k}^{\infty}b_{k,m}q^m.
\ee
Using \eqref{z2k+1} and \eqref{fourierheck}, the  $b_{k,m}$ coefficients are  obtained as follows
\be\label{bkm}
b_{k,m}=\sum_{r=0}^{k}a_{-r}\sum_{d|(r,3m)}\frac{r}{d}c_i\left(\frac{rm}{d^2}\right),\quad i=1,2.
\ee
where $c_1(m)$, and $c_2(m)$ are  the $\mathfrak{j}$  and $\mathfrak{j}^2$ Fourier expansion's coefficients, respectively. The  $c_i$  coefficients are obtained as follows ( up to the exponentially  suppressed terms)
 \cite{Loran:2010bd}, \cite{Alday:2019vdr}:
\begin{eqnarray}\label{cm7}
c_i(m)=2\pi\sqrt{\frac{i/3}{m-i/3}}I_1\left(4\pi\sqrt{\frac{i}{3}(m-\frac{i}{3})}\right),
\end{eqnarray}
and 
\be
a_r=P(r)-P(r-1).
\ee
The partition numbers $P(r)$ are obtained from Peterson-Rademacher expansion:
\be
P(r)=2\pi\left(\frac{1/24}{r-1/24}\right)^{3/4}\sum_{k=1}^{\infty}\frac{1}{k}\mbox{Kl}\left(r-\frac{1}{24},-\frac{1}{24};k\right) I_{3/2}\left(\frac{4\pi}{k}\sqrt{\frac{1}{24}(r-\frac{1}{24})}\right).
\ee
where, Kl(a,b;k) is the Kloosterman sum.  
Using \eqref{bkm}, the microcanonical entropy obtained as follows
\begin{eqnarray}\label{entropy1}
S(k,m)&=&\ln{b_{k,m}}
\end{eqnarray}
We obtain the entropy in the semiclassical limit. In  \eqref{bkm}, the leading term in the large $km$ limit is $d=1$:
\be\label{log}
b_{k,m}=kc(km)+(k-1)c((k-1)m)+(k-2)c((k-2)m)+\cdots,
\ee
Now, we compare the first and the second terms in the Eq.\eqref{log}. Using \eqref{cm7} and the asymptotic behavior of the Bessel Function:
\be
I_1(z)\sim\frac{e^z}{\sqrt{2\pi z}},\quad\quad z\gg 1,
\ee
 we have
\begin{eqnarray}\label{cd23}
\frac{I_1\left(4\pi\sqrt{\frac{k}{3}(m-\frac{1}{3})}\right)}{I_1\left(4\pi\sqrt{\frac{(k-1)}{3}(m-\frac{1}{3})}\right)}\sim e^{2\pi\sqrt{{(m-1/3)}{k}}}.
\end{eqnarray}
As \eqref{cd23} shows the first term in \eqref{log} is the dominant term and the other terms are exponentially small in $\sqrt{\frac{m-1/3}{k}}$, which are important for large values of $\frac{m-1/3}{k}$.

 In the large $m$  limit, the  $c_i(m)$ coefficients are obtained as follows
  \be\label{cm}
\ln c_i(m)=4\pi\sqrt{\frac{i}{3}(m-\frac{i}{3})}-\frac{3}{4}\ln(m-\frac{i}{3})+\frac{1}{4}\ln\frac{i}{3}-\frac{1}{2}\ln 2\quad\quad i=1,2
\ee
 In the large $k$  and $m$,  where $\frac{m}{k}$ is constant, the leading terms in the entropy  comes from $r=k$. Therefore, in the semiclassical limit the entropy is obtained as follows
\begin{eqnarray}\label{entropy}
S(k,m)&=&\ln{b_{k,m}}=\ln{kc_i\left(k(m-\frac{1}{3})\right)}+\cdots\\\nonumber
&=&4\pi\sqrt{\frac{i}{3}(k(m-\frac{i}{3}))}-\frac{3}{4}\ln(m-\frac{i}{3})+\frac{1}{4}\ln k+\frac{1}{4}\ln\frac{i}{3}-\frac{1}{2}\ln 2+\cdots
\end{eqnarray}
the first  term is well known Bekenstein-Hawking entropy which is proportional with the area of the BTZ black hole.  The other terms determine the logarithmic corrections to the entropy which were calculated before\cite{Strominger:1997eq}-\cite{Govindarajan:2001ee}. These logarithmic correction to the entropy  typically appears in the microcanonical entropy. 
Eq \eqref{cd23} shows the portion of  $r=k$ and $r=k-1$ terms. As this equation shows, for large values of $\frac{m-1/3}{k}$, which means the size of the BTZ black  hole is in the order of AdS scale, these terms are exponentially suppressed. The physical interpretation of these terms would be interesting to investigate.

\section{Summary}
In this study we have investigated the 3d quantum gravity and its corresponding CFT.  The equivalence between 3d gravity and Chern-Simons gauge theory, shows that the boundary CFT should be extremal and the partition function  should be holomorphically factorizable and the left and right central charges are $(24k_L, 24k_R)$. Existence of extremal CFT for $k_L, k_R>1$ is still an open problem. There are two approaches to view this subject. In the positive approach we view the current failure of attempts to non- existence proof as an      indication that extremal CFT exist as a mathematical object. While in the negative approach, the lack of current constructed extremal CFTs is viewed as an indication that they do not exist. Additional investigation is required to clarify matter. The values of the gauge couplings have been obtained from the gauge group of the Chern-Simons theory. 

We studied the holomorphic factorized CFT.  From modular invariant of partition function we showed that the chiral and anti-chiral  functions are $S$ invariant and $T$ covariant. Therefore,  the chiral and anti-chiral  functions are modular covariant. For chiral CFT, scaling dimension of the primary fields are integer numbers and the central charges are integer multiple of 8 $(c_L=24k_L, c_R=24k_R$ for $k_L, k_R\in\frac{\mathbb{Z}}{3})$  . We have obtained  the bases for the chiral  functions in-terms of the Klein function and $\mathfrak{j}=J^{\frac{1}{3}}$. It is believed that for $k\in3\mathbb{Z}$ where the chiral  function is modular invariant (therefore the chiral function itself is a modular invariant partition function and can be consider as a corresponding purely chiral CFT ), the chiral extremal CFT is holomorphically dual with chiral gravity  \cite{Li:2008dq, Maloney:2009ck}.

The gauge group of the Chern-Simons theory can be $SO(2,1)\times SO(2,1)$ and its n-fold diagonal cover. From modular invariance of the holomorphically factorizable partition function, we have shown the symmetry group of the theory can be one of two choices: either $SO(2; 1) \times SO(2; 1)$ or its three-fold diagonal cover. Furthermore, we have introduced the generalized Hecke operators (3-Hecke operators) which map the modular covariant functions to modular covariant functions. We have studied the 3d pure gravity and corresponding boundary CFT for  the case where the gauge group of the Chern-Simons theory is 3-fold diagonal cover (i.e.  left and right central charges are integer multiple of eight ). For the case where the central charges are multiple of 24, it has been tried to calculate the partition function of the pure gravity by different methods including summing over the known classical geometries contributions to the partition function, including quantum corrections \cite{Maloney:2007ud, Keller:2014xba} and the Rademacher expansions \cite{Alday:2019vdr}. The resulting partition functions have some unknown physical properties including the negative norm sates. In this case the chiral partition function is modular invariant.  The calculation of the  partition function of the pure gravity for central charges multiple of 8 (where the chiral  function is modular covariant), using the same method as in\cite{Maloney:2007ud}, would be an interesting matter which can be studied in the future. 

 Using the 3-Hecke operators we have obtained the closed-form for the conjectural partition functions for extremal CFT2s, and the corresponding microcanonical entropies, when the chiral central charges are multiples of eight. We have computed subleading corrections to the Beckenstein-Hawking entropy in the bulk gravitational theory with these conjectural partition functions. We showed  the  microcanonical  entropy is equal to the  Bekenstein-Hawking entropy, the logarithmic corrections and some subleading terms which, are important when the size of the BTZ black hole is of the order of the AdS scale. The logarithmic corrections were obtained before for the BTZ black hole, but the subleading terms are new terms which studying their physical interpretation would be interesting 
 \subsection*{Acknowledgments}
    I am grateful to Steven Carlip, David Mc Gady, and Farhang Loran for reading the manuscript and Stevan Carlip, David Mc Gady, Nathan Benjamin and especially referee for their useful comments.

\end{document}